\documentclass[12pt,letterpaper]{article}
\usepackage[left=1in,top=1in,right=1in,nohead]{geometry}
\usepackage{enumerate}
\usepackage{graphicx}
\usepackage{amsmath}
\usepackage{amsthm}
\usepackage{mathrsfs}
\usepackage{amssymb}
\usepackage{multirow}
\usepackage{color}
\usepackage{bm}
\usepackage{comment}
\usepackage{subfigure}
\usepackage{paralist}
\usepackage{pst-coil,pstricks-add,pst-plot}
\usepackage{jheppub}

\theoremstyle{remark}
\newtheorem{thm}{Theorem}
 
\newtheorem{lem}{Lemma}
\newtheorem{defn}{Definition}

\newcommand{\Int}{\mathrm{Int}}
\newcommand{\Ext}{\mathrm{Ext}}
\newcommand{\grad}{\nabla}
\newcommand{\I}{\mathscr{I}}
\newcommand{\Ocal}{\mathcal{O}}
\newcommand{\eps}{\epsilon}

\def\bea#1\eea{\begin{align}#1\end{align}}
\def\be#1\ee{\begin{equation}#1\end{equation}}

\title{Covariant Constraints on Hole-ography}

\author{Netta Engelhardt}
\author{and Sebastian Fischetti}

\affiliation{Department of Physics, University of California, Santa Barbara, CA 93106, USA}

\emailAdd{engeln@physics.ucsb.edu}
\emailAdd{sfischet@physics.ucsb.edu}

\abstract{Hole-ography is a prescription relating the areas of surfaces in an AdS bulk to the differential entropy of a family of intervals in the dual CFT.  In (2+1) bulk dimensions, or in higher dimensions when the bulk features a sufficient degree of symmetry, we prove that there are surfaces in the bulk that cannot be completely reconstructed using known hole-ographic approaches, even if extremal surfaces reach them.  Such surfaces lie in easily identifiable regions: the interiors of holographic screens. These screens admit a holographic interpretation in terms of the Bousso bound.  We speculate that this incompleteness of the reconstruction is a form of coarse-graining, with the missing information associated to the holographic screen.  We comment on perturbative quantum extensions of our classical results.}

\begin{document}

\maketitle
\flushbottom


\section{Introduction}
\label{sec:intro}

The deep connection between entanglement and geometry has the potential to provide profound insights into the inner workings of a nonperturbative theory of quantum gravity. This connection has been made especially manifest in the AdS/CFT duality, which relates certain conformal field theories (CFT) without gravitational dynamics to string theory on asymptotically (locally) anti-de Sitter (AdS) backgrounds~\cite{Mal97,Wit98a}.  In this correspondence, the CFT lives on a representative of the conformal class of boundary metrics of the AdS space; we colloquially say that the CFT ``lives on the boundary of AdS''.  In the limit where the string theory is well approximated by classical gravity, the dual CFT is strongly coupled (large~$\lambda$) with a large number of colors (large~$N$).  Numerous observables in the CFT are dual in this limit to geometric objects in the (now classical) AdS space.

In this context, an issue of considerable interest is that of bulk reconstruction.  That is, given some CFT data, how much of the bulk data can be reconstructed, and how is this reconstruction performed?  Understanding how this reconstruction works in the limit where the AdS bulk is classical may offer insights into how to reconstruct the bulk perturbatively in~$1/N$, and even potentially in a nonperturbative regime (\textit{i.e.}~finite~$N$).

Because many CFT observables are dual to geometric bulk constructs in the large~$N$ limit, a fundamental bulk object to reconstruct is the geometry itself.  A promising approach has focused on reconstructing the bulk using extremal codimension-two surfaces anchored to the boundary: according to the Ryu-Takayanagi (RT) and Hubeny-Rangamani-Takayanagi (HRT) conjectures~\cite{RyuTak06, HubRan07}, such extremal surfaces are dual to the entanglement entropy of regions of the CFT.  In fact, arguments made by~\cite{CzeKar12, Wal12} suggest that the density matrix of a subregion of the CFT should be sufficient to reconstruct a portion of the bulk geometry (the domain of dependence of a set of relevant extremal surfaces~\cite{CzeKar12}, or the region bounded by a null hypersurface fired from an extremal surface~\cite{Wal12}). Indeed,~\cite{Czech:2014ppa,Czech:2015qta} explicitly offer such a construction for the spatial slices of AdS$_3$ by using the hole-ographic approach~\cite{Balasubramanian:2013rqa,Balasubramanian:2013lsa} of reconstructing bulk surfaces from boundary-anchored extremal surfaces (see also~\cite{Czech:2014wka,HeaMye14,Myers:2014jia} for related constructions).

The appeal of this approach stems from its conceptual simplicity: it relates (\textit{a priori}) any bulk surface to CFT observables.  Specifically, the area of an arbitrary bulk surface~$\gamma$ is dual to the so-called differential entropy of a family of boundary intervals.  The full range of validity of hole-ography remains unclear, though substantial headway in this direction was made in~\cite{HeaMye14}.  In this paper, we continue this exploration: in any (2+1)-dimensional spacetime (or in any higher dimensional spacetime with a sufficient degree of symmetry), we will state and prove general theorems that constrain how well surfaces in the bulk spacetime can be reconstructed from extremal surfaces anchored to the AdS boundary.  We interpret these constraints in terms of the so-called holographic screens introduced in~\cite{CEB2}.  We emphasize that while our strongest theorems only apply to systems that are ``effectively (2+1)-dimensional'', they are otherwise covariant.  In particular, while our results are constrained in more than two spatial dimensions to these effectively lower-dimensional setups, in (2+1)-dimensional bulk spacetimes we impose no restrictions except a generic condition and a condition on the Ricci tensor (the null curvature condition), which amounts to positivity of the stress tensor for a bulk obeying the Einstein equation. 

To give these statements some context, recall that the Hubeny-Rangamani-Takayangi (HRT) conjecture~\cite{HubRan07} states that in the large-$N$ limit, the entanglement entropy of a region $\mathcal{R}$ in the CFT can be constructed as follows.  Consider all bulk codimension-two extremal surfaces~$X$ homologous to the region~$\mathcal{R}$ on the AdS boundary\footnote{Note that the homology constraint (see~\textit{e.g.}~\cite{Headrick:2007km}) implies that~$X$ must be anchored to the AdS boundary on~$\partial\mathcal{R}$.}.  Then the entanglement entropy of~$\mathcal{R}$ is
\be
\label{eq:HRT}
S(\mathcal{R}) = \min_{X\sim \mathcal{R}} \frac{\mathrm{Area}(X)}{4G_N \hbar},
\ee
where~$G_N$ is the bulk Newton's constant and~$\sim$ means ``homologous to''.  Both the left- and right-hand sides of the above equation are na\"ively divergent and are understood to be regulated appropriately.  A generalization of this prescription exists for perturbatively quantum bulk spacetimes~\cite{FauLew13, EngWal14}.

\begin{figure}[t]
\centering
\includegraphics[page=1]{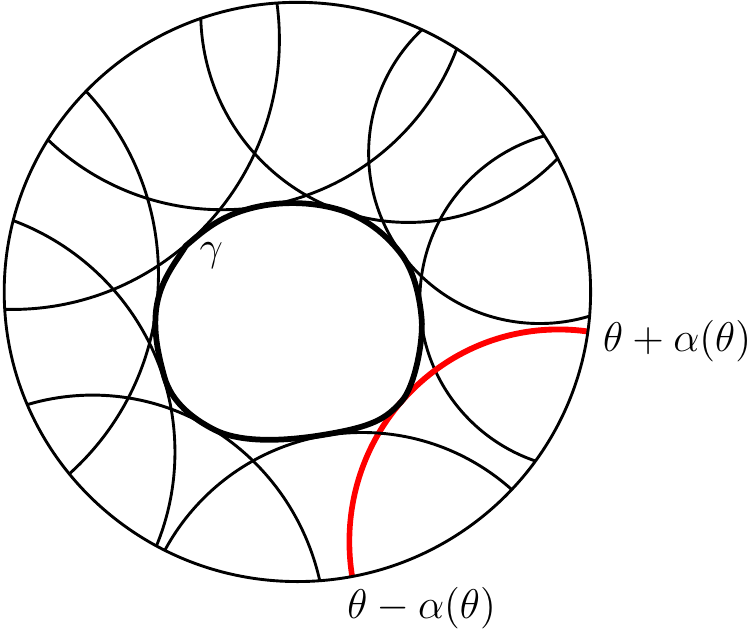}
\caption{An arbitrary closed curve~$\gamma$ on a static time slice of global~AdS$_3$.  The set of all geodesics tangent to~$\gamma$ define a family of regions on the boundary parametrized by a (possibly multi-valued) function~$\alpha(\theta)$.  The differential entropy of these regions gives the length of~$\gamma$.}
\label{fig:holeography}
\end{figure}

The key insight of hole-ography is that the HRT formula~\eqref{eq:HRT} can in certain cases be utilized to compute the area of arbitrary bulk surfaces.  In the pure AdS$_3$ context, consider an arbitrary curve~$\gamma$ lying on a static time slice, as shown in Figure~\ref{fig:holeography}.  At each point~$p$ on~$\gamma$, there is a unique geodesic tangent to~$\gamma$ at~$p$ anchored at the ends of some boundary interval~$I_\theta = (\theta - \alpha(\theta),\theta + \alpha(\theta))$; here information about~$\gamma$ is contained in the region function~$\alpha(\theta)$.  By the RT (and HRT) conjectures, the length of this geodesic computes the entanglement entropy~$S(\alpha)$ of the interval~$I_\theta$.  The result of~\cite{Balasubramanian:2013lsa} is that the length of~$\gamma$ can be computed from the boundary entanglement entropies as
\be
\label{eq:holeography}
\frac{\mathrm{length}(\gamma)}{4G_N \hbar} = \frac{1}{2} \int_0^{2\pi} d\theta \, \left. \frac{dS(\alpha)}{d\alpha} \right|_{\alpha = \alpha(\theta)}.
\ee
This construction has been generalized to non-static contexts and higher dimensions (admitting a sufficient degree of symmetry) in~\cite{HeaMye14,Czech:2014wka}. The quantity on the right-hand side was termed ``differential entropy'' in~\cite{Myers:2014jia}, related but not equivalent to the residual entropy discussed in~\cite{Balasubramanian:2013rqa, Balasubramanian:2013lsa}. In particular,~\cite{Hub14} showed that the reconstruction of bulk curves from the residual entropy is subject to strong restrictions; the differential entropy is not subject to the same constraints~\cite{HeaMye14}.

In order to use the hole-ographic approach for bulk reconstruction,~\cite{Czech:2014ppa} suggested that points in the bulk spacetime can be identified by effectively shrinking~$\gamma$ to arbitrarily small size around a point~$p$, so that the geodesics tangent to~$\gamma$ all intersect at~$p$; see Figure~\ref{subfig:holepoint}.  The resulting region function~$\alpha_p(\theta)$ is an extremum of a boundary action constructed only from~$S(\alpha)$, and thus provides a definition of bulk points from boundary data.  Similarly, to compute the geodesic distance between two points~$p$ and~$q$,~$\gamma$ is shrunk to a thin convex\footnote{In this context, a closed curve~$\gamma$ is convex if any geodesic connecting two points on~$\gamma$ lies entirely inside~$\gamma$.} curve that encircles~$p$ and~$q$, as shown in Figure~\ref{subfig:holedist}.  The region function for such a curve can be constructed from those that define the points~$p$ and~$q$,~$\alpha_p(\theta)$ and~$\alpha_q(\theta)$, and is therefore also constructed purely from boundary data.

\begin{figure}[t]
\centering
\subfigure[]{
\includegraphics[page=2]{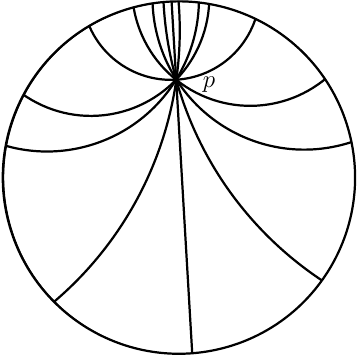}
\label{subfig:holepoint}
}
\hspace{1cm}
\subfigure[]{
\includegraphics[page=3]{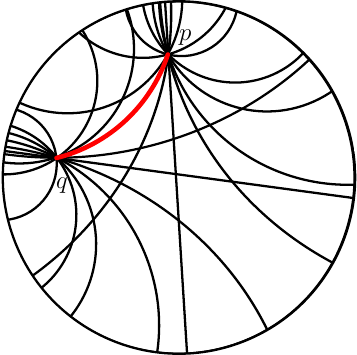}
\label{subfig:holedist}
}
\caption{\subref{subfig:holepoint}: reconstruction of bulk points via hole-ography.  The curve~$\gamma$ is shrunk to be arbitrarily small and centered at~$p$, so that~$p$ is identified by the common intersection of all the geodesics generated by~$\alpha(\theta)$.  \subref{subfig:holedist}: reconstruction of geodesic distances via hole-ography.  The curve~$\gamma$ is shrunk to be an arbitrarily thin convex curve (thick red line) encircling two points~$p$ and~$q$.  The geodesic distance between~$p$ and~$q$ is then given by the differential entropy of the resulting boundary intervals.}
\label{fig:reconstruction}
\end{figure}

This approach is clean and elegant and has close ties to integral geometry~\cite{Czech:2015qta} and to tensor networks and MERA~\cite{Swingle:2009bg,Swingle:2012wq}.  It is therefore quite natural to ask how much it can be generalized, and how much of the bulk it can reconstruct.

One obvious impediment to this reconstruction is the presence of extremal surface barriers (or relatedly, bulk regions that cannot be reached by any HRT surfaces -- ``entanglement shadows''~\cite{Freivogel:2014lja}).  These are surfaces that split the bulk spacetime in two such that no codimension-two extremal surface can be deformed to cross them~\cite{EngWal13}.  Then anything behind the barrier cannot be probed via boundary entanglement entropy.  Interestingly, in~\cite{Hubeny:2012ry} it was found that under appropriate restrictions, extremal surfaces anchored to one asymptotic boundary cannot be deformed to enter the event horizons of static black holes.  This barrier phenomenon was characterized for arbitrary spacetimes in~\cite{EngWal13}; in particular, such barriers do not include the event horizons of dynamical black holes.  Thus generically, an event horizon is a barrier only in stationary setting.

This is not so surprising: in a dynamical context, an event horizon is a global object, but from a local perspective, its only special property is the fact that its area is non-decreasing.  Since extremal surfaces are not sensitive to the global structure of the spacetime, there is no reason to expect the event horizon to generically play a special role in constraining their behavior.  A much more promising alternative is that of \textit{local} analogues of the event horizon: it is common to consider dynamical horizons~\cite{AshKri02} or trapping horizons~\cite{Hay93}, but we will instead consider more general objects called holographic screens~\cite{CEB2}.  These will be defined precisely in Section~\ref{sec:theorems} below, but they should roughly be thought of as objects that can be foliated by marginally trapped (or anti-trapped) surfaces.  Holographic screens can be constructed from an arbitrary foliation of a spacetime\footnote{This means a given spacetime may generally admit infinitely many holographic screens: one per null foliation.}; we will illustrate such a construction in Section~\ref{sec:theorems} (see Figure~\ref{fig:screenconstruction}).

Our motivation for focusing on these screens is fourfold.  First, there is a sense in which they are analogues of event horizons that are local in time and defined independenly of an asymptotic boundary.  Second, it was shown in~\cite{BouEng15a, BouEng15b} that under certain (fairly generic) assumptions, they obey an area law much like that obeyed by event horizons.  Third, they have a holographic interpretation by the Bousso bound~\cite{CEB1}: their area places an upper bound on the total entropy lying on one of the null surfaces orthogonal to them.  The fourth and last point is a technical one: holographic screens can be constructed from a null foliation of spacetime, and null congruences are very useful in constraining the behavior of codimension-two extremal surfaces.  Thus it should be relatively straightforward to derive constraints on such surfaces in the presence of holographic screens.

Interestingly, our results show that while there are indeed such constraints, they are subtle.  Holographic screens need not be barriers: codimension-two extremal surfaces may enter them.  However, we prove that when they do, the extremal surfaces must move through a certain  subregion of the interior of a holographic screen monotonically\footnote{In the special case where the extremal surfaces are anchored on a connected boundary region,  extremal surfaces must move monotonically through the entire interior of the holographic screen.}.  That is, they may never become tangent to one of the leaves of the null foliation that was used to construct the screen.  This puts a limit on how well hole-ographic approaches can reconstruct surfaces and geometry in the interior of a holographic screen, for any (sufficiently smooth) codimension-two spacelike surface~$\gamma$ lying inside the screen must be tangent to at least two of the null foliation surfaces.  This implies that there are points -- and more generally open subsets -- on~$\gamma$ that cannot be tangent to any boundary-anchored codimension-two extremal surface.

Thus we prove a no-go theorem for hole-ography: it cannot be used to reconstruct arbitrary surfaces contained in the interiors of holographic screens.  At best, it can reconstruct only portions of them, yielding some ``coarse-grained'' form of reconstruction.

The outline of this paper is as follows.  We develop and state our main theorems in Section~\ref{sec:theorems}.  In the interest of readability, we will defer the lengthier of our proofs to Appendix~\ref{app:proofs}.  In Section~\ref{sec:examples} we present some examples illustrating the ideas used in our construction, and highlighting previous instances in the literature where hints of our results first appeared.  Finally, in Section~\ref{sec:discussion} we discuss the relevance of our results to bulk reconstruction, as well as some possible generalizations, and conclude.


\section{Constraints on the Behavior of Extremal Surfaces}
\label{sec:theorems}

In this section, we will state the theorems discussed in Section~\ref{sec:intro}.  For pedagogical reasons, some results (specificially Lemma~\ref{lem:NARWHAL} and Theorem~\ref{thm:main}) will be presented for~(2+1) dimensions first.  Section~\ref{subsec:higherD} provides a generalization to higher dimensions.  For this reason, we will continue to discuss ``codimension-two surfaces'' rather than ``curves'', so the generalization to higher dimensions is natural.

Furthermore, while we will narrate the development of the theorems for purposes of pedagogy and clarity, we will leave a discussion and interpretation of their consequences to Sections~\ref{sec:examples} and~\ref{sec:discussion}.  Terms in quotation marks are intended to provide intuition, and will be made precise in due course.

\textbf{Preliminaries}  We will always consider a spacetime~$M$ that obeys the null curvature condition: $R_{ab}k^{a}k^{b} \geq 0$ everywhere for any null vector $k^a$.  Unless otherwise specified, we take all null vectors to be future-directed.  The term \textit{extremal surface} will always be used to refer to spacelike, $C^{2}$, codimension-two extremal surfaces.  A null hypersurface and the null geodesic congruence that generates it will be given the same name (\textit{e.g.}~$N$).  The expansion of a congruence~$N$ will be denoted~$\theta(N)$, while the expansion of a spacetime-filling family of congruences~$\{N_s\}$ will be denoted~$\theta(\{N_s\})$.  All unspecified conventions and definitions are as in~\cite{Wald}.

\subsection{General Behavior of Null Hypersurfaces and Extremal Surfaces}

First, we introduce a null foliation~$\{N_s\}$ of~$M$ into null hypersurfaces~$N_s$ which we shall call \textit{leaves}\footnote{Recall that the leaves~$\{N_s\}$ form a foliation of~$M$ if for every~$p \in M$,~$p$ lies on precisely one leaf~$N_s$.  Also, note that this foliation is arbitrary; any spacetime admits infinitely many such foliations.}. The leaves are permitted to have cusps, but only at intersections of their generators; a generator leaves a leaf if and only if it encounters an intersection with another generator of the same congruence.  Next, recall that any extremal surface~$X$ has two null normals, each of which generates a null congruence (as shown in Figure~\ref{fig:ExtremalCongruences}), and  the extremality condition is simply the requirement that the expansions of the null geodesic congruences tangent to these normals vanish on~$X$.  If~$X$ is tangent to a null hypersurface~$N$, the extremality of~$X$ constrains the expansion of~$N$:

\begin{figure}[t]
\centering
\includegraphics[page=4]{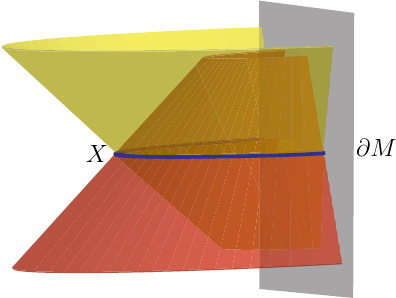}
\caption{The two null congruences of an extremal surface~$X$ of codimension two anchored to a timelike boundary~$\partial M$.}
\label{fig:ExtremalCongruences}
\end{figure}

\begin{lem}
\label{lem:aron}
Let $N$ be a null hypersurface in~$M$ and let $X$ be a codimension-two spacelike extremal surface which is tangent to $N$ at a point $p$; let~$\Ocal_p$ be an open neighborhood of~$p$.  Then:
\begin{itemize}
	\item If $X \cap \Ocal_p$ is nowhere to the past of $N$, then $\left. \theta(N)\right|_{p}\leq 0$;
	\item If $X \cap \Ocal_p$ is nowhere to the future of $N$, then $\left. \theta(N)\right|_{p}\geq 0$.
\end{itemize}
\end{lem}

\begin{proof}
As explained in~\cite{BouEng15b}, this follows directly from Theorem 1 of~\cite{Wal10QST} or Theorem 4 of~\cite{Wal12}.
\end{proof}

As a useful illustration of this lemma, consider extremal surfaces and light cones in flat space, as shown in Figure~\ref{fig:lemma}.

\begin{figure}[t]
\centering
\includegraphics[page=5]{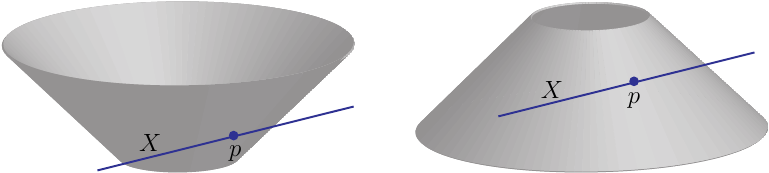}
\caption{An illustration for Lemma~\ref{lem:aron}.  In Minkowski space, an extremal surface~$X$ is just a plane (drawn here as a straight line).  If~$X$ is tangent to an expanding light cone, it lies nowhere to the cone's future, and the cone has positive expansion.  If~$X$ is tangent to a shrinking light cone, it lies nowhere to the cone's past, and thus the shrinking light cone has negative expansion.}
\label{fig:lemma}
\end{figure}

The converse of Lemma~\ref{lem:aron} is in general not true\footnote{We thank Aron Wall for pointing this out to us.}.  However, in the restricted case of a~(2+1)-dimensional spacetime, we can indeed prove its converse  -- see Section~\ref{subsec:higherD} for a generalization to higher dimensions:

\begin{lem}
\label{lem:NARWHAL}
Let $N$ be a null hypersurface in a~(2+1)-dimensional spacetime $M$ and let $X$ be a codimension-two spacelike extremal surface which is tangent to $N$ at a point $p$. Then there exists a small neighborhood $\Ocal_p$ of $p$ such that
\begin{itemize}
	\item If $\left. \theta(N)\right|_{p}> 0$, then $X \cap \Ocal_p$ is nowhere to the future of $N$;
	\item If $\left. \theta(N)\right|_{p}< 0$, then $X \cap \Ocal_p$ is nowhere to the past of $N$.
\end{itemize}
\end{lem}

\begin{proof}
Consider the first case, where the expansion of $N$ is positive.  At $p$, the null generator of $N$ agrees with a null normal of $X$; call this vector $k^a$.  Also let~$v^a$ be the unit vector tangent to~$X$ at~$p$ (which will also be tangent to~$N$, since~$X$ is), and let~$\ell^a$ be the other null normal to~$X$ at~$p$ normalized so~$k \cdot \ell = -1$.  Then the metric at~$p$ can be decomposed as
\be
g_{ab}|_p = -2k_{(a}\ell_{b)} + v_a v_b.
\ee
The expansion of~$N$ at~$p$ can then be written as
\be
\label{eq:expansioncondition}
0 < \theta(N)|_p = \left. \grad_a k^a \right|_p = \left. \ ^{N}\!K_{ab}v^{a}v^{b} \right|_p,
\ee
where $^{N}\!K_{ab}$ is the extrinsic curvature of $N$\footnote{Recall that the extrinsic curvature of a null codimension-one hypersurface with normal~$k^a$ is given (up to scaling) by
\be
K_{ab} = \frac{1}{2} \pounds_k g_{ab}.
\ee
For a codimension-two surface with null normals~$k^a$ and~$\ell^a$, the extrinsic curvature gets an extra index:
\be
K^{a}_{\phantom{a}bc} = \frac{1}{2} \left(\ell^a \pounds_{k} g_{bc} + k^a \pounds_{\ell} g_{bc}\right).
\ee}.
Next, consider a spacelike surface~$\Sigma$ containing~$X$.  Recall that~$^{N}\!K_{ab}v^{a}v^{b}|_p$ is a measure of how much~$N \cap \Sigma$ bends away from its tangent plane (\textit{i.e.}~the plane spanned by~$k^a$ and~$v^a$) with motion away from $p$ in the $v^{a}$ direction.  By extremality, the trace of the extrinsic curvature of $X$ vanishes:~$^{X}\! K^{c}_{\phantom{c}ab}v^{a}v^{b}|_{p}=0$, so $X$ must curve away from its tangent plane less than $N \cap \Sigma$ on a small open neighborhood of $p$. But this immediately implies that $X\cap \Ocal_p$ cannot lie in the future of $N$. The proof proceeds identically for the second case.
\end{proof}

Lemmata~\ref{lem:aron} and~\ref{lem:NARWHAL} give conditions on how extremal surfaces are allowed to be tangent to null hypersurfaces.  Crucially, these conditions do not impose any restrictions on the global structure of the null hypersurface -- it may be a hypersurface of non-constant expansion on a global scale, but as long as it has definite expansion on an open set that contains $p$, both lemmata are applicable.  This means that in any region of the spacetime with constant sign of~$\theta(\{N_s\})$ -- a scalar function on the spacetime -- an extremal surface can ``turn around'' at most once with respect to the foliation~$\{N_s\}$ (this notion will be made precise below).  In order to understand the general behavior of extremal surfaces, it is therefore useful to divide the spacetime into those regions where~$\theta(\{N_s\})$ is positive, and those where~$\theta(\{N_s\})$ is negative.

\subsection{Holographic Screens}

The division between regions of positive and negative~$\theta(\{N_s\})$ is provided quite naturally by so-called \textit{preferred} holographic screens, first defined in~\cite{CEB2}.  The idea is the following: given a spacetime foliation~$\{N_s\}$, move along each leaf~$N_s$ until its expansion changes sign.  By the focusing theorems (see \textit{e.g.}~\cite{Wald}), this sign change can happen at most once (since the expansion of~$N_s$ is non-increasing).  Thus, assuming a generic condition to be stated below, to each leaf~$N_s$ this procedure associates at most one codimension-two surface~$\sigma_s$ called a \textit{leaflet}\footnote{Note that this terminology goes against convention: typically the~$\sigma_s$ are referred to as ``leaves''.  Here we reserve the term ``leaves'' for the null hypersurfaces of the spacetime foliation.}.  The union of all such leaflets is a preferred holographic screen and provides the division we were looking for; see Figure~\ref{fig:screenconstruction} for an example of this construction. The term holographic screen is derived from the Bousso bound, which postulates that the leaflet is holographic: its area provides a bound on the entropy of $N_{s}$ \cite{CEB1, CEB2}.

\begin{figure}[t]
\centering
\includegraphics[page=6]{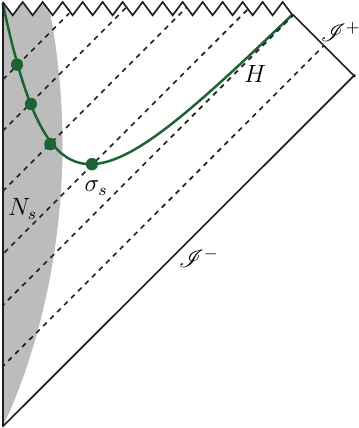}
\caption{Constructing a preferred holographic screen from a null foliation of a spacetime.  The dashed diagonal lines are the leaves of the foliation; the dot on each leaf marks the leaflet~$\sigma_s$ where the expansion of the leaf changes sign.  The union of all the leaflets is a preferred holographic screen.}
\label{fig:screenconstruction}
\end{figure}

Note that each leaflet~$\sigma_s$ has two null normal directions, each tangent to an associated null congruence.  By construction, one of these congruences has zero expansion.  We can use the sign of the expansion of the other congruence to label the ``type'' of holographic screen: in analogy with event and dynamical horizons, a screen will be called ``future'' (``past'') if it is foliated by marginally (anti-)trapped surfaces~\cite{BouEng15a, BouEng15b}.  This notion is made precise by the following definition:

\begin{defn}
\textit{Preferred future holographic screen}. A preferred future holographic screen $H$ associated to a null spacetime foliation $\{N_{s}\}$ is a smooth hypersurface such that for each leaf~$N_s$, the intersection $H\cap N_{s}$ is either empty or a codimension-two achronal surface $\sigma_{s}$ such that the two orthogonal null directions~$k^a_s$ and~$\ell^a_s$ to~$\sigma_{s}$ obey:
\bea
\theta_{k_s} &= 0, \\
\theta_{\ell_s} &< 0,
\eea
where $\theta_{k_s,\ell_s}$ are the expansions of the null geodesic congruences fired off of~$\sigma_s$ in the~$k^a_s$ and~$\ell^a_s$ directions.  The intersections $\sigma_{s}$ are called \textit{leaflets} of $H$, and the null normals~$k^a_s$ and~$\ell^a_s$ to all the leaflets define null vector fields~$k^a$ and~$\ell^a$ everywhere on~$H$.
\end{defn}

Past holographic screens are defined analogously, except that~$\theta_{\ell} > 0$, \textit{i.e.}~the leaflets are marginally \textit{anti}-trapped.  All discussions and proofs for past holographic screens proceed identically to future holographic screens via time reversal (all future constructs become past-directed), so for the rest of this section we will refer only to future holographic screens.

The above definition of holographic screens is too weak to guarantee that they be sufficiently well-behaved for our purposes.  But by further imposing some mild conditions, it is possible to ensure that the screens obey certain ``nice'' properties.  For this reason, we require the screen to be regular:

\begin{defn}
\label{def:regular}
\textit{Regular future holographic screen}.  A preferred future holographic screen is \textit{regular} if the following are true~\cite{BouEng15b}:
\begin{itemize}
	\item The null expansion of leaflets in the $k^a$ direction immediately decreases away from $H$: $k^a_s\grad_a \theta_{k_s}|_{\sigma_s} < 0$;
	\item The boundary of all spacelike subsets of $H$ within $H$ is the boundary of all timelike subsets of $H$ within $H$ (\textit{i.e.}~the only null portions of~$H$ are junctions between spacelike and timelike pieces);
	\item Every inextendible portion of $H$ with indefinite sign is either timelike or contains a complete leaflet; and
	\item Every leaflet is compact and splits a Cauchy surface containing it into two disjoint subsets.
\end{itemize}
\end{defn}

The first two assumptions can be viewed as types of generic conditions\footnote{However, these do not reduce to the usual generic condition used in the singularity theorems, see \textit{e.g.}~\cite{Wald}.}.  We will not have occasion to explicitly use the last two assumptions in this section, but they are required for certain properties of regular holographic screens to hold.  Also note that we will occasionally use the word ``screen'' to refer to a regular holographic screen when it will cause no ambiguity. 

The screens on which we will focus must divide the spacetime into two disjoint regions so that we can sensibly refer to their ``interior'' and ``exterior''.  Such screens will be referred to as \textit{splitting screens}; the holographic screen shown in Figure~\ref{fig:screenconstruction} is an example.  Moreover, if the screen is regular, we can uniquely define its interior and exterior:~\cite{BouEng15a, BouEng15b} showed that when~$k^a$ points to one side a regular screen, it is always the same side ($k^a$ may be tangent to the screen, but never switches from one side to the other).  Thus we will call the interior~$\Int(H)$ of a splitting future holographic screen~$H$ the region towards which the null vector field~$k^a$ points\footnote{This definition may seem backwards, since we typically think of the ``interior'' of a surface as the direction in which the expansion of its null normals is more negative.  However, note that since marginally trapped surfaces must always lie behind (or possibly on) the future event horizon of the spacetime~$M$,~$\Int(H)$ can never have any intersection with the asymptotic region of~$M$.  It is in this sense that this definition agrees with intuition.
}. The exterior~$\Ext(H)$ will be the complement in~$M$.

We are now equipped to make statements about the behavior of extremal surfaces in general spacetimes in the presence of holographic screens.  We need one more definition to make precise what we mean by an extremal surface ``turning around'':

\begin{defn}
\textit{Turning and inflection points}.  We say that an extremal surface $X$ has a \textit{pivot point} at a point $p$ if it is tangent to a leaf~$N_s$ at~$p$.  $X$ is said to have a \textit{turning point} at $p$ if in a small neighborhood of $p$, $X$  lies nowhere to the past or nowhere to the future of~$N_s$.  Otherwise,~$X$ is said to have an \textit{inflection point} at~$p$.  Moreover, if an extremal surface~$X$ has a turning point in some region~$R \subset M$, then we say~$X$ \textit{turns around} in~$R$.  See Figure~\ref{fig:pivot} for an illustration.
\end{defn}

Note that turning points and the notion of turning around are dependent on the foliation~$\{N_s\}$.  Also note that by definition, if $N$ is any null splitting hypersurface, then any surface $Q$ which has a turning point on $N$ is (in some small neighborhood) in its past or future.  In the former case, we will say $Q$ is \textit{tangent to $N$ from the past}, and in the latter we will say $Q$ is \textit{tangent to $N$ from the future}.

\begin{figure}[t]
\centering
\includegraphics[page=7]{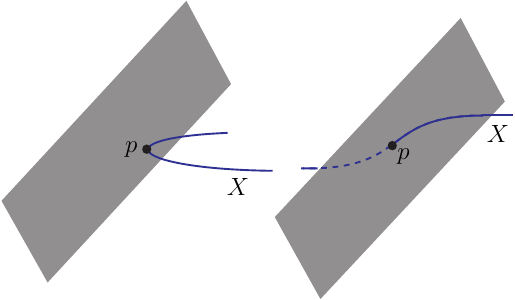}
\caption{An extremal surface tangent to a foliation leaf has a pivot point which can either be a turning point (left) or an inflection point (right).}
\label{fig:pivot}
\end{figure}

\subsection{Theorems}

We can now state our first theorem, which is simply a precise rephrasing of the heuristic discussion above:

\begin{thm}
\label{thm:traffic}
Let $R$ be a region such that~$\theta(\{N_s\})$ has a definite sign everywhere in~$R$, and let~$X$ be a (codimension-two) extremal surface.  Then any connected portion of $X$ in~$R$ can turn around at most once and has no inflection points if~$M$ is~(2+1)-dimensional.  In particular, if~$H$ is a regular splitting future holographic screen, any connected portion of~$X$ in~$\Int(H)$ can turn around at most once. 
\end{thm}

For a detailed proof, see Appendix~\ref{subapp:traffic}.  For a pictoral proof, see Figure~\ref{fig:traffic}: if a connected portion of~$X$ in~$R$ has more than one turning point, at least one such turning point must violate Lemma~\ref{lem:aron}.  

\begin{figure}[t]
\centering
\includegraphics[page=8]{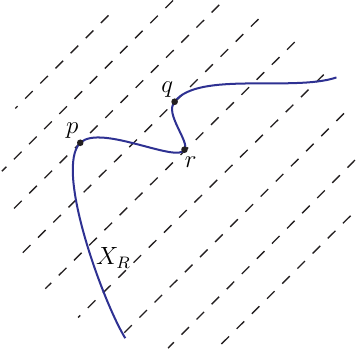}
\caption{$X_{R}$ is any connected portion of~$X$ in~$R$.  It cannot have multiple turning points without violating Lemma~\ref{lem:aron}.}
\label{fig:traffic}
\end{figure}

Theorem~\ref{thm:traffic} and the lemmata make local statements: they make no use of the global structure of the spacetime~$M$.  We now focus on the asymptotically locally AdS case\footnote{See~\cite{Fischetti:2012rd} for the definition of asymptotically locally AdS spacetimes.}, where~$M$ has a timelike boundary~$\partial M$ to which the extremal surfaces are anchored.  Then within a subregion of the interior of a holographic screen, Theorem~\ref{thm:traffic} can be strengthened significantly. We will call this region the \textit{umbral} region\footnote{We will show that~$U(H)$ is similar to but more general than the partial shadows of~\cite{Freivogel:2014lja}, hence the nomenclature.} $U(H)$:  

\begin{defn} \textit{Umbral region}.  Let $H$ be a regular splitting future holographic screen.  Consider the null congruences~$\{L_s\}$ generated from each leaflet by firing null geodesics in the~$\ell^a$ direction\footnote{As with the~$N_s$, we will take the generators of~$L_s$ to leave~$L_s$ at local and non-local intersections.}, and suppose that these foliate~$\Int(H)$. Let $\sigma_{s,s'}=N_{s}\cap L_{s'}$. The umbral region $U(H)$ is the union of all those $\sigma_{s,s'}$ with no intersection with the past of $H$:
\begin{equation}
U(H) \equiv \bigcup \sigma_{s,s'} \ : \ \sigma_{s,s'}\cap I^{-}(H)=\varnothing .
\end{equation} 
See Figure~\ref{fig:umbral} for an illustration.
\end{defn}
Note that if $H$ is achronal, $U(H)$ is the entire interior of $H$.

\begin{figure}[t]
\centering
\includegraphics[page=24]{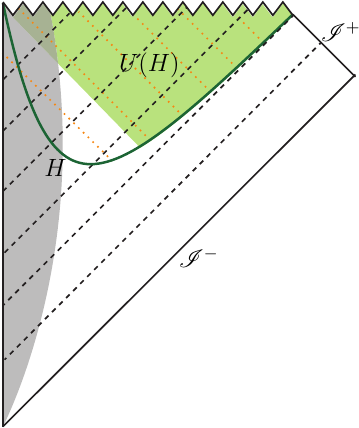}
\caption{In a collapsing star geometry, the future holographic screen $H$ (solid green) has indefinite signature.  Its umbral region~$U(H)$ is the (green) shaded region in the interior of~$H$; by construction, the future of~$U(H)$ has no intersection with~$H$.}
\label{fig:umbral}
\end{figure}

\begin{thm}
\label{thm:main}
Let $H$ be a regular splitting future holographic screen in a~(2+1)-dimensional asymptotically locally AdS spacetime~$M$.  Then no boundary-anchored extremal surface can have a pivot point in~$U(H)$.  In particular, if $H$ is achronal, no such extremal surfaces can have a pivot point in~$\Int(H)$.
\end{thm}

For a detailed proof, see Appendix~\ref{subapp:main}. For a sketch of part of the proof, consider an extremal surface~$X$ with a turning point on some leaf~$N_m$ in~$U(H)$.  If at this turning point~$X$ is also tangent to a leaf~$L_m$, then Lemma~\ref{lem:NARWHAL} implies that~$X$ must lie to the future of~$N_m$ and~$L_m$, and therefore to the future of their intersection; see Figure~\ref{fig:mainsketch}.  But by assumption, this future can have no intersection with $H$, implying that $X$ must live entirely in the interior of $H$; therefore $X$ cannot be boundary-anchored.  The case when~$X$ is not tangent to a leaf~$L_m$ is more complicated, but similar in spirit.

One of the appealing properties of regular future holographic screens found in~\cite{BouEng15a} is that they obey an area law even when they have indefinite signature.  It is therefore natural to ask whether our theorem applies to such screens.  It well may be the case that it does, but the approach used in the proof above cannot be used for non-achronal screens.  To see why, consider Figure~\ref{fig:prooffail}: an extremal surfaces $X$ can have a turning point somewhere in the past of the screen, in which case it may exit the interior of the screen through a timelike portion.

Theorem~\ref{thm:main} may still be true for regular future holographic screens of indefinite signature, but it is not clear to us how a proof of such a statement would proceed.  However, some progress can be made in higher dimensions, as we will now discuss.

\newpage

\begin{figure}[h]
\centering
\includegraphics[page=23]{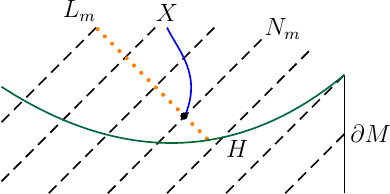}
\caption{If an extremal surface~$X$ (solid blue) is tangent to leaves~$N_m$ (dashed black) and~$L_m$ (dotted orange) inside the umbral region of a holographic screen~$H$ (solid green), it must lie entirely in the future of their intersection (black dot), and therefore cannot be boundary-anchored.  Note that we have suppressed a spatial direction in this figure, which is why~$X$ appears timelike and ends at the black dot.  It is actually spacelike everywhere and tangent to the dot in the suppressed direction.}
\label{fig:mainsketch}
\end{figure}

\vspace{2cm}

\begin{figure}[h]
\centering
\includegraphics[page=12]{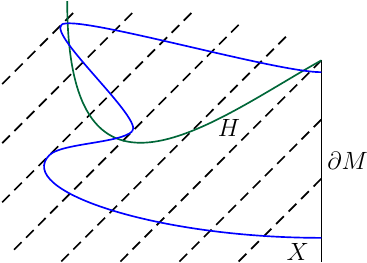}
\caption{Here we show how an extremal surface~$X$ (solid blue) may potentially have a turning point in the interior of a regular future holographic screen of indefinite signature.  Because the turning point lies in the past of the screen, the extremal surface may exit the screen through the timelike portion; none of the turning points shown are forbidden by Lemma~\ref{lem:aron}.  Note that the~$X$ is everywhere spacelike; the apparent mixed signature here is only due to the suppression of the extra spatial dimension.}
\label{fig:prooffail}
\end{figure}

\newpage

\subsection{Higher Dimensions}
\label{subsec:higherD}

Theorem~\ref{thm:main} relies heavily on Lemma~\ref{lem:NARWHAL}, which is only valid in~(2+1)-dimensional spacetimes.  If we were to attempt a na\"ive extension of it to~$(d+1)$ dimensions, we would encounter a problem: equation~\eqref{eq:expansioncondition} and the extremality condition would become
\begin{subequations}
\label{eqs:higherd}
\bea
0 &< \ ^N \! K_{ab} \left. \left(v^a v^b + \sum_i \xi^a_{(i)} \xi^b_{(i)}\right)\right|_p, \\
0 &= \ ^X \! K^a_{\phantom{a}bc} \left. \left(v^b v^c + \sum_i \xi^b_{(i)} \xi^c_{(i)}\right)\right|_p,
\eea
\end{subequations}
where the sum runs over an additional~$(d-2)$ orthonormal spatial vectors~$\xi^a_{(i)}$ that are orthogonal to~$v^a$ and tangent to~$X$ and~$N$ at~$p$.  These summed expressions do not allow us to separately compare the bending of~$X$ and~$N$ in different directions, so the proof does not go through as it did before.

It is clear from the above considerations, however, that a version of Lemma~\ref{lem:NARWHAL} will remain true in higher dimensions if we require that all of the~$\xi^a_{(i)}$ have trivial contraction with the extrinsic curvatures of~$X$ and~$N$.  In such a case, only the~$v^a$ terms in equations~\eqref{eqs:higherd} remain, and the proof of the lemma proceeds as in the~(2+1)-dimensional case.  We therefore define:

\begin{defn}
\textit{Reducibility to~(2+1) dimensions}.  Let~$S$ be a surface of codimension at most two in a~$(d+1)$-dimensional spacetime~$M$.  We will say that~$S$ is \textit{reducible to~(2+1) dimensions} (or \textit{reducible} for short) if there exist~$(d-2)$ vector fields~$\xi_{(i)}^a$,~$i = 1,\ldots,d-2$ in~$M$ that are everywhere spacelike\footnote{The~$\xi_{(i)}^a$ may be singular on a set of measure zero.} such that~$S$ is tangent to the~$\xi^a_{(i)}$ everywhere, and for each~$i$
\be
\ ^S K^a_{\phantom{a}bc} \, \xi^b_{(i)} \, \xi^c_{(i)} = 0,
\ee
where~$\ ^S K^a_{\phantom{a}bc}$ is the extrinsic curvature of~$S$.  If multiple reducible surfaces share the~$\{\xi^a_{(i)}\}$, then they are \textit{simultaneously} reducible.
\end{defn}

For an arbitrary surface~$S$, the reducibility condition is simply a constraint on its allowed behavior.  However, we will specifically require that extremal surfaces be reducible: this will in general only be possible when the spacetime exhibits a sufficient amount of symmetry. In particular, note that in spacetimes obeying the generalized planar symmetry of~\cite{HeaMye14} (see their Section~3.3 for a full definition), extremal surfaces that have this symmetry are reducible.  For example, spacetimes with planar or spherical symmetry provide a simple setup admitting reducible extremal surfaces.  Most significantly,  Lemma~\ref{lem:NARWHAL} still holds for reducible extremal surfaces and foliations, and therefore so does Theorem~\ref{thm:main}.

We have therefore shown that \textit{Theorems~\ref{thm:traffic} and~\ref{thm:main} will hold in any~($d$+1)-dimensional spacetime if the foliations~$\{N_s\}$ and all extremal surfaces~$X$ under consideration are simultaneously reducible to~(2+1) dimensions}.

As mentioned at the end of the previous section, we can actually do more in higher dimensions: for~$d > 2$, it is possible for extremal surfaces to ``cap off'' (for example, where the size of spheres spanned by the~$\xi^a_{(i)}$ shrinks to zero).  

Before we state and prove the theorem, however, we require a restricted notion that takes into account the global structure of the extremal surface. Specifically, we restrict our analysis to so-called \textit{H-deformable} extremal surfaces~\cite{EngWal13}: these are surfaces that can be deformed to lie entirely in the exterior of a screen~$H$ while still being kept extremal (for a precise definition, see Appendix~\ref{app:proofs}).  We therefore have the following theorem, which holds for the entire interior of an arbitrary regular holographic screens (\textit{i.e.}~it is not restricted to the umbral region):

\begin{thm}
\label{thm:mainconnected}
Let~$M$ be an asymptotically locally AdS spacetime, and let $H$ be a regular splitting future holographic screen constructed from a reducible foliation~$\{N_s\}$.  Assume that that there exists a foliation of the future of~$H$ with $L_s$ congruences, which are simultaneously reducible with the $\{N_{s}\}$ leaves.  Let $X$ be a boundary-anchored, codimension-two spacelike extremal surface such that:
\begin{enumerate}
	\item $X$ is reducible to~(2+1) dimensions simultaneously with $\{N_{s}\}$ and $\{L_{s}\}$;
	\item $\partial X$ is connected; and
	\item $X$ intersects~$\Ext(H)$ only on regions with $\theta(\{N_s\}) > 0$.
\end{enumerate}
Assume further that there exists an~$H$-deformation of~$X$ that obeys the above conditions as well. Then~$X$ cannot have a pivot point in~$\Int(H)$.
\end{thm}

Note that condition~(2) rules out geodesics, so this theorem is only nontrivial in~$d > 2$.  Also, condition~(3) is meant to exclude possible pathological behavior from other holographic screens somewhere else in the spacetime.

For a detailed proof of this theorem, see Appendix~\ref{subapp:mainconnected}.  For a rough picture, consider the case where an extremal surface~$X$ is not tangent to an~$L_s$ leaf at its turning point, as in Figure~\ref{fig:prooffail}.  In such a case,~$X$ must be anchored on the boundary at two places, \textit{i.e.}~$\partial X$ is disconnected.  The case where~$X$ is tangent to an~$L_s$ leaf at its turning point is more subtle and requires invoking~$H$-deformability; we leave the details to the Appendix.

It is worth making some remarks about potential pitfalls in higher dimensional spacetimes in which the extremal surfaces and/or null foliations are not reducible.  As noted above, Lemma~\ref{lem:NARWHAL} will then generally be false, and cannot be used to rule out inflection points.  We suspect it should be possible to use only Lemma~\ref{lem:aron} to prove weaker versions of Theorems~\ref{thm:main} and~\ref{thm:mainconnected} that do not exclude inflection points.  However, such constraints have minimal relevance for hole-ography.

We should also note that while our proofs do not hold in non-reducible settings, we can think of no counterexamples to the statements of the theorems.  It is possible that they hold in more generality, but if that is the case, they would need to be proven using a different approach than that taken here.


\section{Examples}
\label{sec:examples}

Here we present examples illustrating the application and consequences of the theorems discussed in the previous section.

\subsection{dS and AdS Spacetimes}

\begin{figure}[t]
\centering
\subfigure[]{
\includegraphics[page=15]{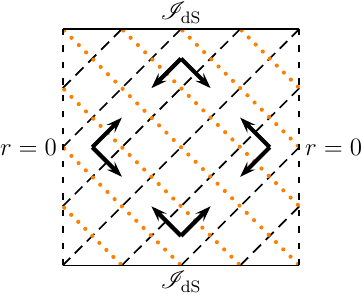}
\label{subfig:dS}
}
\hspace{1cm}
\subfigure[]{
\includegraphics[page=16]{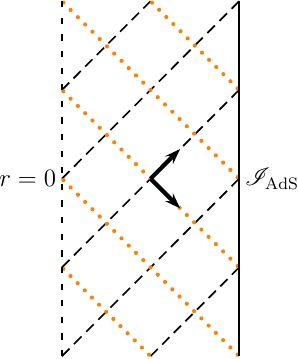}
\label{subfig:AdS}
}
\caption{The conformal diagrams of de Sitter \subref{subfig:dS} and anti-de Sitter \subref{subfig:AdS} space.  Each point on these diagrams corresponds to a suppressed sphere~$S^{d-1}$ whose area is parametrized by a radial coordinate~$r$.  The null foliations shown are generated by light rays fired from~$r = 0$, \textit{i.e.}~the north and south poles of dS and the origin of AdS.  The black arrows indicate the directions in which extremal surface are allowed to turn around (\textit{e.g.}~an arrow pointing down and to the right indicates that extremal surfaces may only be tangent to the dashed foliation from the past).  They imply that extremal surfaces must bend away from~$\I_\mathrm{dS}$, but towards~$\I_\mathrm{AdS}$.}
\label{figs:dSAdS}
\end{figure}

As an example of Theorem~\ref{thm:traffic} (which states that connected components of extremal surfaces can have no more than one turning point in a region of constant~$\theta(\{N_s\})$), consider the simplest cases of pure de Sitter (dS) or anti-de Sitter (AdS) spacetimes, whose conformal diagrams are shown in Figure~\ref{figs:dSAdS} (the analysis of Minkowski space is similar to that of AdS, so we will not discuss it separately).  Both dS and AdS have a spherical isometry to which we have adapted the conformal diagrams; we introduce a coordinate~$r$ which parametrizes the areas of the spheres of symmetry\footnote{Specifically,~$r$ is the usual radial coordinate that appears in the slicing
\be
ds^2 = -\left(1 \pm r^2/\ell^2\right) dt^2 + \frac{dr^2}{1 \pm r^2/\ell^2} + r^2 d\Omega^2_{d-1},
\ee
with the positive (negative) sign for the global (static) slicing of AdS (dS).}.

In each spacetime we introduce two null foliations which we take to be adapted to its spherical isometry: these foliations are generated by light cones fired from~$r = 0$ towards the boundary~$r = \infty$.  It is then easy to use Theorem~\ref{thm:traffic} to understand how extremal surfaces must behave.  The cross-sectional area of the null foliations increases with~$r$, so the expansion along each foliation is positive in the direction of increasing~$r$.  It then follows that the expansion along each sheet of the foliations never changes sign.  This allows us to draw on Figure~\ref{figs:dSAdS} the directions in which extremal surfaces are allowed to turn with respect to these foliations.  In particular, note that extremal surfaces in AdS must bend \textit{towards} the conformal boundary~$\I_\mathrm{AdS}$, while extremal surfaces in dS sufficiently near the boundary~$\I_\mathrm{dS}$ must bend \textit{away} from it.

In principle, these claims only constrain the behavior of extremal surfaces with respect to the two null foliations introduced here.  However, the high degree of symmetry of both dS and AdS allows us to conclude that \textit{all} extremal surfaces in AdS must be attracted to~$\I_\mathrm{AdS}$, while \textit{all} extremal surfaces in dS must be repelled from~$\I_\mathrm{dS}$.  The former point is, of course, well-known: extremal surfaces anchored to the boundary of AdS come up frequently in holographic contexts, and necessarily bend towards the boundary.  The latter point was made generally in~\cite{Fischetti:2014uxa} using similar considerations to the ones used here.  In particular, it follows that no boundary-anchored extremal surfaces exist in dS, since they would necessarily need to bend towards the boundary.

\subsection{AdS Spherical Collapse}

The above simple examples of dS and AdS illustrate how Theorem~\ref{thm:traffic} puts constraints on the general behavior of extremal surfaces in arbitrary spacetimes, even those not containing splitting holographic screens.  Our focus, however, is on applications to AdS/CFT and bulk reconstruction.  To that end, let us now discuss how Theorem~\ref{thm:main} (which states that bounday-anchored extremal surfaces in the interior of achronal screens cannot have pivot points) explains some of the observations of~\cite{Liu:2013iza,Liu:2013qca, Hubeny:2013dea} in the context of null collapse in AdS.

To briefly review, consider the formation of a black hole in Poincar\'e AdS by infalling null dust.  In the holographic context, this process is dual to the thermalization of the boundary field theory following a perturbation (typically a form of a quantum quench).  The bulk solution consists of two pieces: to the past of the null dust, the solution is a vacuum solution and therefore just (the Poincar\'e patch of) pure AdS.  The portion of the bulk containing the dust and to the future of it is AdS-Vaidya:
\be
ds^2 = -f(r,v) dv^2 + 2 \, dv \, dr + \frac{r^2}{\ell^2} \, d\vec{x}_{d-1}^2,
\ee
where
\be
f(r,v) = \frac{r^2}{\ell^2}\left(1 - \frac{\mu(v)}{r^d}\right),
\ee
$d$ is the boundary spacetime dimension, and we can think of compactifying the planar directions~$\vec{x}$ into a torus (it is the planar symmetry of these directions that allows us to apply Theorem~\ref{thm:main} here).  Here the mass function~$\mu(v)$ characterizes the profile of the dust; the null energy condition is satisfied when~$\mu'(v) \geq 0$.  The full solution is shown in Figure~\ref{subfig:AdSVaidyathick}.

\begin{figure}[t]
\centering
\subfigure[]{
\includegraphics[page=17]{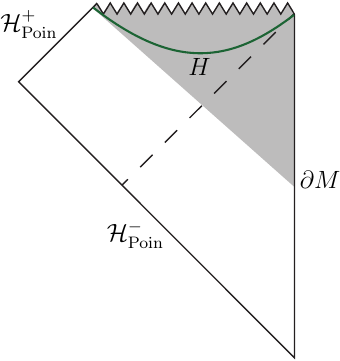}
\label{subfig:AdSVaidyathick}
}
\hspace{1cm}
\subfigure[]{
\includegraphics[page=18]{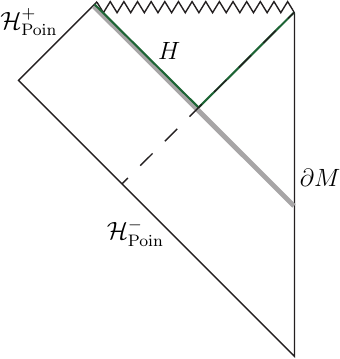}
\label{subfig:AdSVaidyathin}
}
\caption{The formation of a black hole in pure AdS by infalling null dust (the shaded gray regions).  To the past of the dust, the solution is pure AdS; the portion of the spacetime containing the dust is AdS-Vaidya.  \subref{subfig:AdSVaidyathick}: for continuously infalling null dust, the spacetime contains an achronal future holographic screen (solid green curve).  \subref{subfig:AdSVaidyathin}: if the dust is taken to be a thin shell, the screen approaches two null pieces, with one lying on the event horizon and the other on the shell.  However, if the shell has an arbitrarily small but nonzero thickness, and if an arbitrarily small but nonzero amount of matter continues to fall in after the shell, the screen will be achronal (and arbitrarily close to being null).  Note that the null boundaries are actually the Poincar\'e horizons~$\mathcal{H}^\pm_\mathrm{Poin}$.}
\label{fig:nullcollapse}
\end{figure}

Let us now consider the plane symmetric foliation of this spacetime generated by light cones fired from~$r = 0$.  The cross-sectional areas of these sheets go like~$A \propto r^{d-1}$, so the expansion is positive in the direction of increasing~$r$.  In particular, this means that the expansion of the right-moving null sheets to the future of the event horizon changes sign, giving rise to a future holographic screen.  This screen coincides with the dynamical horizon at~$f(r,v) = 0$.

In the context of holographic quantum quenches,~\cite{Liu:2013iza,Liu:2013qca} considered such a collapse scenario with the infalling null matter taken to be a thin shell.  The resulting holographic screen technically violates the assumptions of our theorems, since it is null and therefore not regular.  However, it is easy to consider a regulated solution in which the null shell is smeared out slightly and given a rapidly decaying tail all the way into the far future.  The screen will then be slightly deformed into a regular achronal screen, as illustrated in Figure~\ref{subfig:AdSVaidyathin}.  Then our theorems can be applied to the regulated collapse geometry.  By taking the limit where the regulator goes to zero, we may expect our theorems to apply to the thin shell solution as well.

Theorem~\ref{thm:main} asserts that any reducible boundary-anchored extremal surface cannot turn around inside the screen.  This is precisely what~\cite{Liu:2013iza,Liu:2013qca} found, as shown in Figure~\ref{fig:LiuSuh}.  Extremal surfaces anchored to strips on the boundary sometimes penetrate the screen, but the turning point never does.  In particular, as the boundary strips are taken to later times, the turning point of the corresponding extremal surfaces tracks out a curve which never enters the screen.  Thus the interesting behavior of the turning point shown in Figure~\ref{fig:LiuSuh} is simply a consequence of our theorem.

\begin{figure}[t]
\centering
\includegraphics[page=19]{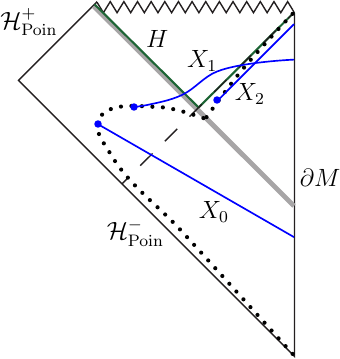}
\caption{A family of extremal surfaces (solid blue lines) anchored to the boundary of AdS in thin shell Vaidya-AdS, as found in~\cite{Liu:2013iza,Liu:2013qca}.  The dot at the end of each surface indicates the location of its turning point; the dotted black line follows the path of this turning point as time at which the surface is anchored to~$\partial M$ is varied.  Note that some surfaces in this family do enter the holographic screen~$H$, but the turning points never do.}
\label{fig:LiuSuh}
\end{figure}

Ref.~\cite{Hubeny:2013dea} considered a similar problem, but in global AdS.  In that case, the generalized planar symmetry is (a subset of) the full spherical symmetry. They too found extremal surfaces anchored to spherical boundary regions that penetrated the holographic screen, but never any that turned around in it.


\section{Ramifications for Hole-ography}
\label{sec:discussion}

The key questions of hole-ography are: does there exist an object in the CFT which is dual to the area of an arbitrary spacelike codimension-two bulk surface?  If so, what are the limitations of this duality?  The former question has been addressed in~\cite{Balasubramanian:2013rqa,Balasubramanian:2013lsa,Czech:2014wka,HeaMye14,Myers:2014jia}; in this paper, we have proven theorems that give a partial answer to the latter.

\subsection{Incomplete Reconstruction Inside Screens}

Recall that~\cite{HeaMye14} showed that under an appropriate set of assumptions (including generalized planar symmetry), if a given bulk spacelike codimension-two surface~$\gamma$ can be reconstructed from boundary-anchored extremal surfaces tangent to it, then the area of~$\gamma$ is given by the differential entropy of the boundary regions selected by the extremal surfaces.  This direction was referred to as the ``bulk-to-boundary'' direction.  Conversely, given a set of intervals on the boundary, the extremal surfaces anchored to them can be used to define at least one bulk surface~$\gamma$ whose area is equal to the differential entropy of the intervals; this is the ``boundary-to-bulk'' direction.
  
We  pause here to note  an important subtlety: to get a good correspondence between the area of~$\gamma$ and the CFT differential entropy, the extremal surfaces must be the minimal-area ones that are picked out by the HRT formula (since there may exist more than one surface with the same boundary conditions).  More generally, if the extremal surfaces used to reconstruct~$\gamma$ are not the minimal-area ones, they may be related to other CFT quantities such as entwinement~\cite{Balasubramanian:2014sra} (for example, minimal surfaces alone cannot be used to reconstruct the AdS$_3$ conical defect geometry or BTZ~\cite{Czech:2014ppa}).  Here we will show that surfaces inside holographic screens cannot be fully reconstructed from \textit{any} boundary-anchored extremal surfaces, be they minimal-area or not.

\begin{figure}[t]
\centering
\subfigure[]{
\includegraphics[page=20]{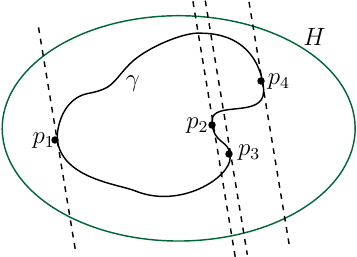}
\label{subfig:gammainH}
}
\hspace{1cm}
\subfigure[]{
\includegraphics[page=21]{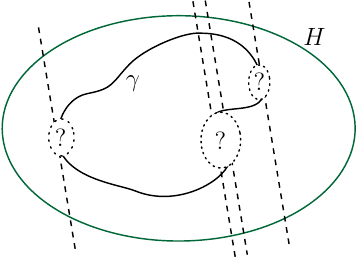}
\label{subfig:gammainHreconstruct}
}
\caption{The plane of the page is a time slice containing a spacelike curve~$\gamma$ (solid black line) in the interior of a holographic screen~$H$; the green oval shows the intersection of~$H$ with this particular time slice.  \subref{subfig:gammainH}: the curve~$\gamma$ will always be tangent to at least two leaves of the foliation (dotted black lines); in this particular case, it is tangent to four of them at the marked points.  \subref{subfig:gammainHreconstruct}: by our theorem,  portions of~$\gamma$ in a neighborhood of these points cannot be reconstructed from boundary-anchored extremal surfaces.}
\label{fig:curveinH}
\end{figure}

For example, consider the consequences of our results for a bulk-to-boundary construction: let~$\gamma$ be a sufficiently smooth spacelike closed curve\footnote{Here we will restrict the discussion to three bulk dimensions (so~$\gamma$ is just a curve), though our statements also hold in reducible setups.} that lies entirely in the interior of some regular holographic screen~$H$.  Since~$\gamma$ is smooth, there must be some points at which~$\gamma$ is tangent to leaves of the null foliation used to construct~$H$.  We have illustrated this in Figure~\ref{subfig:gammainH}, where we have shown a spatial slice containing~$\gamma$ and its intersection with~$H$ and some leaves of the null foliation.  Theorem~\ref{thm:main} implies that there cannot exist boundary-anchored geodesics tangent to~$\gamma$ at the marked points (any extremal surfaces tangent to~$\gamma$ there must \textit{e.g.}~end at a singularity).  Moreover, if we slightly deform the null foliation, these points will shift slightly along~$\gamma$, so we find that there are open regions of~$\gamma$ to which no boundary-anchored geodesics are tangent.

This is our main result:~$\gamma$ cannot be entirely reconstructed from any set of boundary-anchored geodesics, minimal or not.  Generically, however, there will be regions of~$\gamma$ that \textit{can} be.  Thus in this bulk-to-boundary approach,~$\gamma$ can only be \textit{partially} reconstructed from boundary-anchored geodesics (and therefore in principle from CFT observables dual to them).  This is a form of coarse-graining: the boundary data dual to geodesics simply do not know how to reconstruct some pieces of~$\gamma$.  This coarse-grained reconstruction is illustrated in Figure~\ref{subfig:gammainHreconstruct}.

Recall, however, that the boundary-to-bulk approach of~\cite{HeaMye14} is slightly different: in order to reconstruct (the area of) a bulk curve~$\gamma$ from a set of boundary intervals, we do not need the corresponding geodesics to be tangent to~$\gamma$.  Rather, we only require what~\cite{HeaMye14} call the ``null alignment condition'': where a geodesics meets~$\gamma$, their tangent vectors need not agree, but may simply span a null plane.  This is a weaker constraint, and it is therefore natural to wonder if the boundary-to-bulk construction fares any better in this case.

The answer is no.  Suppose a smooth bulk curve~$\gamma$ constructed via the boundary-to-bulk approach is contained entirely inside~$H$.  Consider the two null planes generated by congruences fired off of~$\gamma$ in its four orthogonal (past and future) null directions.  The null alignment condition says that~$\gamma$ may be constructed from boundary-anchored geodesics that intersect~$\gamma$ and are tangent to one of these planes when they do.  But since~$\gamma$ is smooth, by the same argument given above there must exist some points at which this null plane is tangent to a leaf of the foliation.  Then proceeding as we did in the bulk-to-boundary construction, we conclude that~$\gamma$ must contain segments that cannot be constructed from boundary-anchored geodesics.

The conclusion is that whether one takes the bulk-to-boundary or boundary-to-bulk approach, it is not possible to reconstruct an entire smooth\footnote{There may still be reconstruction issues even if~$\gamma$ has cusps, but we will not consider this case here.} curve~$\gamma$ contained inside a holographic screen from boundary-anchored geodesics. In fact, it is very plausible that there are curves of which only an arbitrarily small portion can be reconstructed.

Of course, there is nothing preventing either approach from reconstructing a bulk curve that is only partly contained inside the holographic screen.  However, a promising approach of hole-ography was to be able to reconstruct the bulk geometry itself via the integral geometry approach of~\cite{Czech:2014ppa,Czech:2015qta}.  In order to use this approach to reconstruct the spacetime inside a holographic screen, one would need to reconstruct arbitrary curves entirely contained within it.  It would thus appear that this approach to reconstructing the interior of a holographic screen will not succeed.

\subsection{Quantum Effects}

A possible objection to our conclusion is the following: why not use boundary-anchored extremal surfaces to reconstruct the geometry of a portion of a Cauchy slice~$\Sigma$ to the past of the screen, and then use the bulk equations of motion to evolve forward from~$\Sigma$ to reconstruct its entire causal development?  In particular, this may include the interior of a holographic screen; such an example is shown in Figure~\ref{fig:EOMreconstruction}.

\begin{figure}[t]
\centering
\includegraphics[page=22]{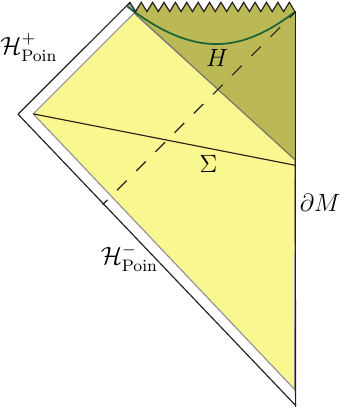}
\caption{Attempting to reconstruct the interior of a holographic screen by evolving forward from an initial time slice~$\Sigma$.  The shaded yellow region shows the domain of dependence~$D(\Sigma)$ of~$\Sigma$; in principle, if we knew the equations of motion everywhere, we could reconstruct this entire domain just from data on~$\Sigma$.  In particular, this can include the interior of a holographic screen~$H$ (green).}
\label{fig:EOMreconstruction}
\end{figure}

In principle this is possible, but only if we know the equations of motion \textit{a priori}.  However, there is a sense in which a ``full'' bulk reconstruction should reconstruct the equations of motion as well as the geometry \textit{ab initio}.  This is especially relevant given that the interiors of holographic screens tend to contain singularities, that is, regions where quantum gravitational effects become important.  As soon as quantum fluctuations are introduced into the metric, even perturbatively, the possibility of reconstructing the bulk from its equations of motion is lost, particularly in near-singularity regions.  For this reason, we find it more natural to seek a way of reconstructing the bulk \textit{directly} from CFT data, without recourse to any equations of motion.

While our work has hithero been entirely classical, the appearance of quantum effects motivates the following observations:
\begin{itemize}
\item  Recall that the interior of a holographic screen has a holographic interpretation in terms of bulk entropy via the Bousso bound~\cite{CEB1, CEB2}.  The area of a leaflet of a future or past holographic screen gives a bound on the entropy of the leaf generating it:
\be
S(N_{s}) \leq \frac{\mathrm{Area}(\sigma_{s})}{2G_{N}\hbar}.
\ee
This raises an interesting question: can the holographic screen itself be reconstructed from boundary observables?  More precisely, what is the CFT dual of a holographic screen, and how is it linked to bulk entropy? As discussed in the previous subsection, the holographic screen is an obstacle to complete hole-ographic reconstruction of its interior; perhaps the information that is lost in the ``coarse-graining'' discussed above is stored in extra degrees of freedom associated with the screen (similar to \textit{e.g.}~the superselection sectors of~\cite{MarWal12}). 
\item Even in the presence of quantum effects, the option of direct reconstruction from boundary observables remains: for a semiclassical bulk (\textit{i.e.}~working to first order in $G_{N}\hbar/\Lambda^{d-1}$ where $\Lambda$ is a characteristic length scale of the quantum fields in the theory),~\cite{FauLew13} found that the \textit{generalized entropy} of extremal surfaces yields the dual CFT entanglement entropy.  More precisely, the generalized entropy of a spacelike codimension-two surface $X$ is given by~\cite{Bek73}:
\be
S_{\mathrm{gen}} (X)= \frac{\mathrm{Area}(X)}{4G_{N}\hbar} + S_{\mathrm{ent}} + \mathrm{counterterms},
\ee
where $S_{\mathrm{ent}}$ is the von Neumann entropy of the exterior of $X$ on some Cauchy surface.  It was later conjectured by~\cite{EngWal14} that, at any finite order in perturbation theory in $G_{N}\hbar/\Lambda^{d-1}$ in the bulk,  there exists a quantum analogue of a classical extremal surface, obtained by replacing the area by the generalized entropy in the extremization procedure.  The quantum extremal surface is obtained by extremizing $S_{\mathrm{gen}}$ with respect to variations along a null surface fired from $X$. The entanglement entropy of the boundary region enclosed by $\partial X$ is conjectured to be dual to the generalized entropy of $X$~\cite{EngWal14}. The extension of the hole-ographic construction to semiclassical and perturbatively quantum gravity has not been discussed, as it is yet to be well-understood even at the classical level.  However, it is very tempting to hope that a similar construction can be made using quantum extremal surfaces.
\item Since quantum fields may violate the null energy condition (which was assumed for all of the proofs in this paper), it may \textit{prima facie} appear that our results are applicable exclusively to the classical case, where reconstruction may be undertaken via the bulk equations of motion. However, it should be possible to prove similar statements about bulk reconstruction from ``quantum hole-ography'' by: (1) replacing all surfaces with their quantum analogues; (2) relinquishing the null energy condition in favor of the recent quantum focussing conjecture of~\cite{BouFis15}, which asserts that the variation of the generalized entropy (rather than the area) is nonincreasing, or equivalently the second variation is nonpositive; and (3) imposing an analogous generic condition to be introduced in~\cite{BouEngToAppear}. In other words, quantum extremal surfaces cannot be used to reconstruct surfaces in spacetime regions foliated by leaves with decreasing generalized entropy. 
\end{itemize}

\acknowledgments

It is a pleasure to thank Raphael Bousso, Dalit Engelhardt, William Kelly, Matthew Headrick, Gary Horowitz, Don Marolf, Aron Wall, and Jason Wien for helpful discussions.  NE is supported by the National Science Foundation Graduate Research Fellowship under Grant No. DGE-1144085.  SF is supported by the National Science Foundation under grant number PHY12-05500 and by funds from the University of California.

\appendix

\section{Proofs}
\label{app:proofs}

In this Appendix, we prove the theorems stated in the main text and provide some more technical details.

\subsection*{Theorem~\ref{thm:traffic}}
\label{subapp:traffic}

\textit{Theorem 1.}
Let $R$ be a region such that~$\theta(\{N_s\})$ has a definite sign everywhere in~$R$, and let~$X$ be a (codimension-two) extremal surface.  Then any connected portion of $X$ in~$R$ can turn around at most once, and has no inflection points if~$M$ is~(2+1)-dimensional.  In particular, if~$H$ is a regular splitting future holographic screen, any connected portion of~$X$ in~$\Int(H)$ can turn around at most once.

\begin{proof}
First, let us index the leaves of the foliation~$\{N_s\}$ by a parameter $s$ which runs to the future along the foliation, \textit{i.e.}~$N_s$ is nowhere to the past of $N_{s'}$ if and only if $s>s'$, which we will also denote by $N_{s}>N_{s'}$.

Now, we prove the theorem by contradiction. Let $\theta(\{N_{s}\})<0$ everywhere in $R$; the opposite case proceeds analogously. Let $X_{R}$ be a connected component of $X$ in $R$, and suppose $X_{R}$ has a pivot point at $p\in R$. Let $N_{s(p)}$ be the leaf containing $p$. By Lemma~\ref{lem:NARWHAL}, $p$ cannot be an inflection point if~$M$ is~(2+1)-dimensional.

Now suppose~$p$ is a turning point.  Suppose also that $X_{R}$ has another turning point $q\in R$; \textit{i.e.}~$X_{R}$ is tangent to another leaf $N_{s(q)}$, and it must be tangent to it either from the past or from the future. By Lemma~\ref{lem:aron}, $X_{R}$ must be tangent to $N_{s(q)}$ from the future. This immediately requires $X_{R}$ to have another turning point $r$ such that $N_{s(r)}>N_{\text{max}(s(p),s(q))}$, and $X_{R}$ must be tangent to $N_{s(r)}$ from the past at $r$, as shown in Figure~\ref{fig:traffic}. But if $r\in R$, then by construction $\theta(N_{s(r)})<0$, and $X_{R}$ cannot be tangent to $N_{s(r)}$ from the past. Therefore $r \notin R$, and $X_{R} \nsubseteq R$, in contradiction with the definition of $X_{R}$.

To prove the last statement of the theorem, we simply note that by construction,~$\theta(\{N_s\})$ has the same sign everywhere in~$\Int(H)$.
\end{proof}

\subsection*{Theorem~\ref{thm:main}}
\label{subapp:main}

\textit{Theorem 2.} Let $H$ be a regular splitting future holographic screen in a~(2+1)-dimensional asymptotically locally AdS spacetime~$M$.  Then no boundary-anchored extremal surface can have a pivot point in~$U(H)$.  In particular, if $H$ is achronal, no such extremal surfaces can have a pivot point in~$\Int(H)$.

\begin{proof}
By contradiction.  Let $X$ be a boundary-anchored extremal surface which has a pivot point $p$ in $U(H)$. By Theorem~\ref{thm:traffic}, $X$ has a turning point at $p$. Let $L_{m'}$ and $N_{m}$ be the leaves of their respective foliations containing $p$. There are two cases:~(1) $X$ is tangent to a leaf $L_{m'}$, or~(2) $X$ is not tangent to $L_{m'}$. We consider the two cases separately.

\textbf{Case 1}: By Lemma~\ref{lem:NARWHAL}, $\Ocal_{p}\cap X\subset J^{+}(L_{m'})$.  We also have that~$\Ocal_{p} \cap X \subset J^+(N_m)$, and thus since $L_{m'}$ and $N_m$ are boundaries of the future of $\sigma_{mm'}$, $\Ocal_{p} \cap X\subset J^{+}(\sigma_{mm'})$.  By the definition of the umbral region,  $J^{+}(\sigma_{mm'})$ has no intersection with~$\Ext(H)$, implying that it is foliated by $N_{s}$ leaves with negative expansion.  Therefore $I^{+}(\sigma_{mm'})\cap X$ can never be tangent to an $N_{s}$ or an $L_{s'}$, since if it were, it would be tangent to them from the past (see the proof of Theorem~\ref{thm:traffic}), in violation of Lemma~\ref{lem:aron}. Thus $X \subset J^{+}(\sigma_{mm'})$. But because trapped and marginally-trapped surfaces always lie to the future of the future event horizon, $\partial I^{-}(\partial M)$, $J^{+}(\sigma_{mm'}) \cap\partial M=\varnothing$. This would imply that $X$ cannot be boundary-anchored, in contradiction with its definition.

\textbf{Case 2}: Since $X$ intersects $L_{m'}$ at $p$ and is not tangent to it there, there exists a small neighborhood of $p$ on which $X$ intersects both $I^{+}(L_{m'})$ and $I^{-}(L_{m'})$.  This immediately implies that there is a small open subset of $X$ which lies in $J^{+}(\sigma_{mm'})$. By assumption, $\partial X\subset \partial M$, and so $\partial X\cap J^{+}(\sigma_{mm'})=\varnothing$. Therefore, there must exist some point~$q \in U(H)$ at which~$X$ is tangent from the past to either an $L_{s'}$ or $N_{s}$ leaf (for if there were not,~$X$ would need to have a boundary in~$J^+(\sigma_{mm'})$).  But in~$\Int(H)$, the expansions of all of the $L_s$ and the $N_{s}$ leaves are negative, so an extremal surface can only be tangent to them from the future.  We have therefore arrived at a contradiction.
\end{proof}

\subsection*{Theorem~\ref{thm:mainconnected}}
\label{subapp:mainconnected}

\begin{defn}
\textit{$H$-deformability}.  Let $\{X_{\alpha}\}$ be a family of boundary-anchored extremal surfaces such that every surface in $\{X_{\alpha}\}$ can be continuously deformed via other surfaces in $\{X_{\alpha}\}$ to some initial surface $X_{0}$ that lies entirely in~$\Ext(H)$. Then every surface in $\{X_{\alpha}\}$ is said to be \textit{$H$-deformable}, and $\{X_{\alpha}\}$ is an \textit{$H$-deformable family}. 
\end{defn}

\noindent \textit{Theorem 3.}
Let~$M$ be an asymptotically locally AdS spacetime, and let $H$ be a regular splitting future holographic screen constructed from a reducible foliation~$\{N_s\}$.  Assume that that there exists a foliation of the future of~$H$ with $L_s$ congruences, which are simultaneously reducible with the $\{N_{s}\}$ leaves.  Let $X$ be a boundary-anchored, codimension-two spacelike extremal surface such that:
\begin{enumerate}
	\item $X$ is reducible to~(2+1) dimensions simultaneously with $\{N_{s}\}$ and $\{L_{s}\}$;
	\item $\partial X$ is connected; and
	\item $X$ intersects~$\Ext(H)$ only on regions with $\theta(\{N_s\}) > 0$.
\end{enumerate}
Assume further that there exists an~$H$-deformation of~$X$ that obeys the above conditions as well. Then~$X$ cannot have a pivot point in~$\Int(H)$.

\begin{proof}
By contradiction.  Let $X_1$ be an $H$-deformable extremal surface with connected~$\partial X_1$.  Assume that~$X_1$ has at least one pivot point $p_1$ in~$\Int(H)$, and let $N_1$ be the leaf containing $p_1$. Consider a deformation parametrized by~$\alpha$ along an $H$-deformable family $\{X_{\alpha}\}$ to a surface $X_0\in \mathrm{Ext}(H)$. Note that $X_0$ exists by definition of $H$-deformability.

Next, let $X_m$ be the last surface in the deformation which is tangent to a leaf $N_m$ at some point $p_m \in \mathrm{Int}(H)\cup H$ (it may be the case that~$X_m = X_0$).  Because the expansion is negative on all leaves inside $H$, if $p_m\in \Int(H)$ we have by Theorem~\ref{thm:traffic} that~$p_m$ is a turning point.  Then evolving backward along the deformation from $X_1$ to $X_0$, there exists some $\epsilon>0$ such that the surface $X_{m-\epsilon}$ must also be tangent to a leaf in Int$(H)$ (since~$\Int(H)$ is open). But this contradicts the definition of $X_m$. So $p_m \notin \Int(H)$, and thus $p_m\in H$. Then there are two cases:~(1) $X_m$ is tangent to $H$ at $p_m$, or~(2) $X_m$ is not tangent to $H$ at $p_m$. We consider the two cases separately.

\textbf{Case 1}: If $X_m$ is tangent to $H$ at $p_m$, then $X_m$ is also tangent to a leaflet $\sigma_m$, and therefore to $L_m$.  By Lemma~\ref{lem:NARWHAL}, $\Ocal_{p_m}\cap X_m\subset J^{+}(L_m)$.  But the expansion of the portion of~$N_m$ to the future of~$L_m$ is negative:~$\theta(N_m \cap J^+(L_m)) < 0$, even if~$\theta(N_m)|_{p_m} = 0$.  Thus although~$X_m$ is tangent to~$N_m$ at a point where~$\theta(N_m) = 0$, with motion away from $p$ along $X_{m}$ is into the region of $\theta(\{N_{s}\})$; see Figure~\ref{subfig:XonH}.  Then by the same reasoning that led to Lemma~\ref{lem:NARWHAL}, it must also bend into~$J^+(N_m)$, and therefore~$\Ocal_{p_m} \cap X_m\subset J^{+}(\sigma_m)$.

\begin{figure}[t]
\centering
\subfigure[]{
\includegraphics[page=13]{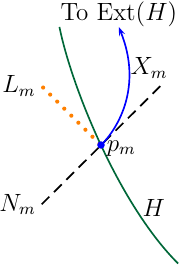}
\label{subfig:XonH}
}
\hspace{2cm}
\subfigure[]{
\includegraphics[page=14]{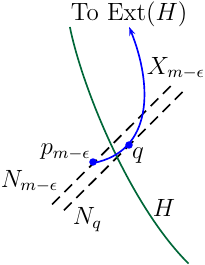}
\label{subfig:XonHpert}
}
\caption{\subref{subfig:XonH}: the behavior of the midway surface~$X_m$ (solid blue) when it is tangent to~$N_m$ and~$H$ (solid green) at a point~$p_m$ on~$H$.  Note here that the surface caps off smoothly in the suppressed spatial directions at $p$; in this sketch, this is shown as the surface ending there. \subref{subfig:XonHpert}: after a small perturbation through the~$H$-deformable family, the turning point~$p_m$ moves out to~$p_{m-\eps}$.  This requires the surface to develop a new turning point~$q$, as shown.  But~$q$ is not allowed to exist. (Note that here we show~$H$ as timelike in a neighborhood of~$p_m$, but the behavior of~$X$ is the same for other signatures).}
\label{fig:XonH}
\end{figure}

If~$X_m \cap J^{+}(\sigma_m)$ has no intersection with~$H$, then we obtain a contradiction as we did for Theorem~\ref{thm:main}.  If~$X_m \cap J^{+}(\sigma_m)$ does have an intersection with~$H$,~$X_m$ must exit~$\Int(H)$ through that intersection.  But now consider a slight deformation to~$X_{m-\eps}$ along the~$H$-deformable family.  Then the pivot point~$p_m$ must deform to a new pivot point~$p_{m-\eps} \in \Ocal_{p_m}$ of~$X_{m-\eps}$ which lies in~$\Ext(H)$ (since by assumption~$X_m$ was the last extremal surface in this family with a turning point in or on~$H$).  Because Int$(H)$ is open, we can always find a sufficiently small deformation from $X_{m}$ to $X_{m-\epsilon}$ such that~$X_{m-\eps}$ must still enter~$\Int(H)$ before exiting, as shown in Figure~\ref{subfig:XonHpert}.  This implies that there must exist another pivot point~$q \in \Ocal_{p_m}$ where~$X_{m-\eps}$ is tangent to a leaf~$N_q$ from the future.  Now, this pivot point cannot lie in~$\Ext(H)$, since there~$\theta(N_q) > 0$.  But this point also cannot lie in~$\Int(H) \cup H$, since by assumption~$X_m$ was the last surface with a pivot point in~$\Int(H) \cup H$.  We therefore have a contradiction.

\textbf{Case 2}: Next, suppose $X_m$ is not tangent to $H$ at $p_m$.  Then as in the proof of Theorem~\ref{thm:main}, there is a small open subset~$\Ocal_{p_m} \cap X_m^+$ of $X_m$ which lies in $I^{+}(L_m)$.  Likewise, there is a small open subset~$\Ocal_{p_m} \cap X_m^-$ of~$X_m$ which lies in~$I^-(L_m)$.  By the arguments made in Case~1, we have that~$\Ocal_{p_m} \cap X_m^+ \subset J^+(\sigma_m)$, and~$\Ocal_{p_m} \cap X_m^- \subset J^-(\sigma_m)$.  Thus~$\Ocal_{p_m} \cap X_m^+$ and~$\Ocal_{p_m} \cap X_m^-$ can meet only on~$\sigma_m \subset N_m$.  In particular,~$\sigma_m$ divides~$X_m$ into two pieces~$X_m^+$ and~$X_m^-$.

Near~$p_m$,~$X_m^-$ lies to the past of~$N_m$.  Therefore, if it were to be tangent to any other leaf~$N_s$, it would have to be tangent from the future.  This is not allowed in~$\Ext(H)$, since there only turning points from the past are allowed, nor is it allowed in~$\Int(H)$ (by assumption).  Therefore~$X_m^-$ must reach the boundary.

Similarly, near~$p_m$,~$X_m^+$ lies to the future of~$L_m$.  By assumption, the~$\{L_s\}$ foliate~$J^+(H)$, so~$X_m^+$ can only have a turning point with respect to~$\{L_s\}$ outside of~$J^+(H)$.  But the fact that near~$p_m$~$X_m^+$ also lies to the future of~$N_m$ implies that~$X^+$ can only leave~$J^+(H)$ if it turns around with respect to~$\{N_s\}$ on some leaf~$N_s > N_m$ (and if it does, it will be tangent to~$N_s$ from the past). But then it can have no further turning points with respect to~$\{N_s\}$:  if it did, these would be from the future. Such turning points cannot occur in~$\Ext(H)$ by Lemma~\ref{lem:NARWHAL} and are not allowed in~$\Int(H)$ by the assumption that~$X_m$ is the last surface to have a turning point in~$\Int(H) \cup H$. Thus~$X_m^+$ must also reach the boundary.

But if each of~$X_m^+$ and~$X_m^-$ reach the boundary, and they join only at~$\sigma_m$, then~$\partial X$ must consist of (at least) two disconnected pieces, in contradiction with the assumption that it be connected.
\end{proof}

\bibliographystyle{jhep}
\bibliography{all}

\end{document}